\documentclass[11pt]{article}
\usepackage{amsmath,amssymb,bbm,amsthm}
\usepackage{fullpage}
\usepackage{thm-restate,color,xcolor,xspace}
\usepackage{hyperref,cleveref}
\usepackage{graphicx}
\usepackage{algorithm,algorithmic}

\newtheorem{theorem}{Theorem}
\newtheorem{lemma}[theorem]{Lemma}
\newtheorem{definition}[theorem]{Definition}
\newtheorem{corollary}[theorem]{Corollary}

\newcommand{\bu}{\mathbf u}
\newcommand{\bg}{\mathbf g}
\newcommand{\by}{\mathbf y}

\begin{document}

\title{Balancing Weights, Directed Sparsification, and Augmenting Paths}
\author{Jason Li\footnote{Carnegie Mellon University. \tt jmli@cs.cmu.edu}}
\date{\today}
\maketitle

\begin{abstract}
We present a randomized augmenting paths-based algorithm to compute the maximum flow in a directed, uncapacitated graph in almost $m+nF$ time, matching the algorithm of Karger and Levine for undirected graphs (SICOMP 2015).
Combined with an initial $\sqrt n$ rounds of blocking flow to reduce the value of $F$, we obtain a maximum flow algorithm with running time $mn^{1/2+o(1)}$. For combinatorial, augmenting paths-based algorithms, this is the first improvement over Dinic's algorithm for moderately sparse graphs.

To obtain our algorithm, we introduce a new technique to re-weight the edges of a strongly connected directed graph so that each cut is approximately balanced: the total weight of edges in one direction is within a constant factor of the total weight in the other direction. We then adapt Karger and Levine's technique of sampling edges from the newly weighted residual graph, ensuring that an augmenting path exists in the sampled graph with high probability. One technical difficulty is that our balancing weights have to be dynamically maintained upon changes to the residual graph. Surprisingly, we can black box the dynamic data structure from the recent interior point method-based flow algorithm of van den Brand et al.\ (FOCS 2024).
\end{abstract}

\section{Introduction}

In the (uncapacitated) maximum flow problem, we are given a directed graph and two vertices $s$ and $t$, and the goal is to find the maximum number of edge-disjoint paths from $s$ to $t$. It is one of the oldest problems in combinatorial optimization~\cite{dantzig1951application}, and many influential techniques have been developed to tackle the problem, including augmenting paths~\cite{ford1956maximal}, blocking flows~\cite{dinic1970algorithm}, push-relabel~\cite{goldberg1988new}, and random sampling~\cite{karger2015fast}. More recently, techniques based on continuous optimization have become widely successful, culminating in the first almost-linear time algorithm for maximum flow and its generalizations~\cite{chen2025maximum}. However, it remains an open question whether a more traditional approach based on augmenting paths can match this recent breakthrough. Two very recent lines of work have obtained almost quadratic-time algorithms for maximum flow~\cite{bernstein2024maximum,bernstein2025combinatorial} and bipartite matching~\cite{chuzhoy2024faster,chuzhoy2024maximum}, showing that combinatorial techniques still have room for future development.

In this work, we present a new approach based on sampling edges from an appropriately weighted residual graph. Our main result is the following:
\begin{theorem}\label{thm:main}
There is a randomized, augmenting paths-based algorithm that computes the maximum flow on a directed, uncapacitated graph in time $(m+nF)\cdot n^{o(1)}$, where $F$ is the value of the maximum flow.
\end{theorem}
Combined with an initial $\sqrt n$ rounds of blocking flow to ensure that $F\le m/\sqrt n$, we obtain the following corollary:
\begin{corollary}\label{cor:main}
There is a randomized, augmenting paths-based algorithm that computes the maximum flow on a directed, uncapacitated graph in time $mn^{1/2+o(1)}$. 
\end{corollary}

For augmenting paths-based algorithms, this is the first improvement over Dinic's algorithm \cite{dinic1970algorithm} for moderately sparse graphs, which runs in time $O(m\cdot\min\{m^{1/2},n^{2/3}\})$ as analyzed by Karzanov~\cite{karzanov1973finding} and Even and Tarjan~\cite{even1975network}.

\subsection{Our Techniques}

Our starting point is Karger and Levine's algorithm on \emph{undirected} graphs~\cite{karger2015fast}, which finds each augmenting path in a random sample of the residual graph. At a high level, Karger and Levine show that if the algorithm has sent a $(1-\epsilon)$-fraction of the maximum flow, then the residual graph is $O(1/\epsilon)$-balanced in the following sense: for any cut, viewed as a bipartition $(S,V\setminus S)$ of the vertex set, the number of edges going from $S$ to $V\setminus S$ is within a multiplicative $O(1/\epsilon)$-factor of the number of edges going from $V\setminus S$ to $S$. Then, they employ the random sampling method of Bencz\'ur and Karger for undirected graphs~\cite{benczur2015randomized}, accounting for the balance of $O(1/\epsilon)$ by over-sampling at a rate of $O(1/\epsilon)$. Overall, they show that sampling $\tilde O(n/\epsilon)$\footnote{$\tilde O(\cdot)$ ignores factors poly-logarithmic in $n$.} edges is sufficient with high probability\footnote{\emph{With high probability} means with probability $1-1/n^c$ for arbitrarily large constant $c$.}, so the next augmenting path can be found in time $\tilde{O}(n/\epsilon)$. Since most of the augmenting paths are found when $\epsilon$ is still constant, the amortized running time per augmenting path is $\tilde{O}(n)$, establishing their running time of $\tilde{O}(m+nF)$.

Unfortunately, the balance guarantee no longer holds if the input graph is directed. For example, suppose the input graph has an unbalanced cut $(S,V\setminus S)$ with $s\in S$ and $t\in V\setminus S$, with only a single edge going from $S$ to $V\setminus S$ but many edges in the opposite direction. Since Karger and Levine's algorithm treats the residual graph as undirected for the purpose of sampling, the single edge from $S$ to $V\setminus S$ cannot be distinguished from the edges in the opposite direction, and is therefore unlikely to be sampled. If the edge is not sampled, then there cannot be a path from $s$ to $t$, invalidating the approach.

Our main idea is to re-weight the residual graph so that any cut is balanced in terms of weight: the total weight of edges in one direction is within a constant factor of the total weight in the other direction. In the example above, the single edge from $S$ to $V\setminus S$ must have much higher weight than the edges in the opposite direction. We then treat the graph as undirected but \emph{weighted} when sampling edges, so that the single edge from $S$ to $V\setminus S$ is likely to be sampled.

One technical difficulty is that our balancing weights have to be dynamically maintained upon changes to the residual graph. More precisely, when an augmenting path is found, the residual graph is updated by flipping each edge along the path. Since Karger and Levine treat the graph as undirected for the purpose of sampling, flipping an edge does not change their sampling probabilities. However, our balancing weights are sensitive to edge flips, and a dynamic algorithm is required to maintain them. Surprisingly, we show that the \emph{minimum ratio cut} framework from the recent interior point method-based flow algorithm of van den Brand et al.~\cite{van2024almost} can be adapted to compute balancing weights. In particular, we completely black box their data structure for maintaining approximate minimum ratio cuts. We remark that this data structure is purely combinatorial and is comparable in complexity to the data structures in~\cite{chuzhoy2024faster,chuzhoy2024maximum}.

\subsection{Related Work}

The concept of \emph{balance} for directed graphs is introduced by Ene et al.~\cite{ene2016routing}, who show that existing techniques for undirected flow-based problems can be translated to directed graphs that are approximately balanced. For graphs that are \emph{perfectly} balanced, namely the Eulerian graphs, Cohen et al.~\cite{cohen2018solving} use directed graph sparsification to obtain fast algorithms for solving directed Laplacian systems. Sparsification for approximately balanced graphs is systematically studied by Cen et al.~\cite{cen2021sparsification}, who obtain tight upper and lower bounds for various types of sparsification. All of these results are restricted to balanced graphs, and to the best of our knowledge, our work is the first to employ techniques from balanced graphs in the context of general directed graphs.

\subsection{Preliminaries}

Throughout the paper, all graphs are uncapacitated, either weighted or unweighted, and either directed or undirected. We refer to directed edges as \emph{arcs} and undirected edges as \emph{edges}. The input directed graph to maximum flow is assumed to be $G=(V,E)$, with vertices $s,t\in V$. We define a \emph{flow} from $s$ to $t$ as a set of arc-disjoint paths from $s$ to $t$ in $G$, and the \emph{value} of the flow is the number of paths in the set. The \emph{residual graph} $G_f$ is obtained from $G$ by flipping each arc along each path in $f$. The standard augmenting paths framework attempts to find a path from $s$ to $t$ in $G_f$, called an \emph{augmenting path}. If a path is found, then each arc along the path is flipped, and the flow $f$ is updated accordingly. Otherwise, if no path exists, then $f$ is a maximum flow, and the algorithm terminates.

A \emph{cut} is a bipartition $(S,V\setminus S)$ of the vertex set into two non-empty subsets, represented as one of the sides $S\subseteq V$. The \emph{boundary} $\partial S$ of the cut is the set of arcs/edges with exactly one endpoint in $S$. For a directed graph, the out-boundary $\partial^+S$ is the set of arcs with only the tail in $S$, and the in-boundary $\partial^-S$ is the set of arcs with only the head in $S$. For a weighted graph with weight function $w$, let $w(\partial S)$ denote the total weight of arcs/edges in $\partial S$, and define $w(\partial^+S)$ and $w(\partial^-S)$ similarly. For an unweighted graph, let $|\partial S|$ denote the number of arcs/edges in $\partial S$, and define $|\partial^+S|$ and $|\partial^-S|$ similarly.

\subsection{Preprocessing the Graph}

We first preprocess the graph as follows. First, add $m$ many parallel arcs $(t,s)$, which does not affect the maximum flow from $s$ to $t$. Then, compute the strongly connected components of the new graph. If $s$ and $t$ do not belong to the same strongly connected component, then the existence of arc $(t,s)$ implies that there is no path from $s$ to $t$, and the algorithm concludes that the maximum flow from $s$ to $t$ has value $0$. Otherwise, discard all strongly connected components except the one containing $s$ and $t$, and solve the maximum flow problem on the remaining component. This operation is safe because there is no path from $s$ to $t$ that contains vertices outside the strongly connected component containing $s$ and $t$.

Let $G$ be the resulting graph after preprocessing. The lemma below ensures that any residual graph $G_f$ remains strongly connected as long as the current flow is not maximum.

\begin{lemma}
Let $f$ be a flow from $s$ to $t$ in $G$ that does not use any arc $(t,s)$. If there is a path from $s$ to $t$ in $G_f$, then $G_f$ is strongly connected.
\end{lemma}
\begin{proof}
It suffices to show that $\partial^+_{G_f}S\ne\emptyset$ for any cut $S\subseteq V$. By assumption, all arcs $(t,s)$ in $G$ remain in $G_f$, and there is a path from $s$ to $t$ in $G_f$. If $t\in S$ and $s\notin S$, then the arcs $(t,s)$ are in $\partial^+_{G_f}S$, so $\partial^+_{G_f}S\ne\emptyset$. If $s\in S$ and $t\notin S$, then a path from $s$ to $t$ in $G_f$ must contain an edge in $\partial^+_{G_f}S$, so $\partial^+_{G_f}S\ne\emptyset$. Finally, if $s,t\in S$ or $s,t\notin S$, then by the conservation of flow, we have $|\partial^+_{G_f}S|=|\partial^+_GS|$, and since $G$ is strongly connected by preprocessing, we conclude that $|\partial^+_{G_f}S|=|\partial^+_GS|\ne0$.
\end{proof}

Our algorithm always finds augmenting paths that avoid the arcs $(t,s)$, so while the current flow is not maximum, we may assume that $G_f$ is strongly connected.

\section{Balancing Weights}

In this section, we construct our balancing weights on the residual graph and describe a data structure to maintain them over the course of the algorithm. The weights assigned to the arcs are determined by a potential function $\by:V\to\mathbb R$ on the vertices. For a given potential function $\by$, the weight of arc $(u,v)$ is
\[ w(u,v)=\frac1{\max\{\by(v)-\by(u),\,0\}+1} .\]
In particular, we have $w(u,v)=\frac1{\by(v)-\by(u)+1}$ whenever $\by(v)\ge\by(u)$, and $w(u,v)=1$ whenever $\by(v)\le\by(u)$.

The algorithm tries to find a potential function $\by:V\to\mathbb R$ such that all cuts are approximately balanced in the two directions: for any set $S\subseteq V$, we require that $\frac12w(\partial^-S)\le w(\partial^+S)\le2w(\partial^-S)$. Actually, we only need to enforce $\frac12w(\partial^-S)\le w(\partial^+S)$ because the other inequality follows by replacing $S$ with $V\setminus S$. If there is a violating set $S$ with $w(\partial^+S)<\frac12w(\partial^-S)$, then the algorithm increases the potential slightly for all vertices in $S$, i.e., sets $\by(v)\gets\by(v)+\Delta$ for all $v\in S$, where $\Delta>0$ is small enough. By a potential function argument, this process must terminate eventually. In particular, a balancing function $\by$ exists, which is already a non-trivial statement.

To ensure a fast algorithm, we need to find a balancing function $\by$ quickly, and moreover, dynamically maintain the function upon finding augmenting paths. We interpret the balance of a directed graph as the \emph{minimum ratio cut} of the underlying \emph{undirected} graph, and then use the dynamic minimum ratio cut data structure from a recent interior-point method-based flow algorithm~\cite{van2024almost}.

We first define the minimum ratio cut of an undirected graph, whose input includes a \emph{gradient function} $\bg:V\to\mathbb R$ with $\bg\perp\mathbbm 1$, i.e., $\sum_{v\in V}\bg(v)=0$.

\begin{definition}
The \emph{minimum ratio cut} of a positively weighted, undirected graph with gradient function $\bg\perp\mathbbm 1$ is the minimum ratio $\bg(S)/w(\partial S)$ over all $S\subseteq V$ with $w(\partial S)\ne0$.
\end{definition}

Observe that the requirement $\bg\perp\mathbbm 1$ means that $\bg(S)+\bg(V\setminus S)=0$ for any $S\subseteq V$, so either $S$ or $V\setminus S$ has non-positive ratio. Intuitively, the goal is to maximize $|\bg(S)|$ while minimizing $w(\partial S)$, and then take either $S$ or $V\setminus S$, whichever has non-positive ratio.

We construct our minimum ratio cut instance as follows. Consider the \emph{undirected} graph with edge weights $w(u,v)$ together with the gradient function $\bg:V\to\mathbb R$ defined as
\[ \bg=\sum_{(u,v)}w(u,v)\cdot(\mathbbm 1_u-\mathbbm 1_v) ,\]
i.e., for each arc $(u,v)$ in the residual graph, add $w(u,v)$ to the tail $u$ and subtract $w(u,v)$ from the head $v$. In other words, we ``undirect'' the residual graph and encode the arc directions through a gradient function.

\begin{lemma}
For any set $S\subseteq V$, we have $\bg(S)=w(\partial^+S)-w(\partial^-S)$.
\end{lemma}
\begin{proof}
Any arc $(u,v)\in\partial^+S$ contributes $w(u,v)$ to $\bg(S)$ since we add $w(u,v)$ to the tail $u\in S$, while the subtraction at the head $v\notin S$ is not counted in $\bg(S)$. Similarly, any arc $(u,v)\in\partial^-S$ contributes $-w(u,v)$ to $\bg(S)$. All other arcs contribute $0$ since either $u,v\notin S$ in which case there is no contribution, or $u,v\in S$ in which case the contributions of $+w(u,v)$ and $-w(u,v)$ cancel out.
\end{proof}

\begin{lemma}\label{lem:approx-balanced}
If the minimum ratio cut has ratio $\bg(S)/w(\partial S)\ge-1/3$, then $w(\partial^+S)\ge\frac12w(\partial^-S)$ for all $S\subseteq V$.
\end{lemma}
\begin{proof}
It suffices to show that a cut $S\subseteq V$ with $w(\partial^+S)<\frac12 w(\partial^-S)$ implies that $\bg(S)/w(\partial S)<-1/3$. By the previous lemma, we have $\bg(S)=w(\partial^+S)-w(\partial^-S)$. Together with $w(\partial S)=w(\partial^+S)+w(\partial^-S)$, we obtain
\begin{align*}
3\cdot\bg(S)+w(\partial S)&=3\cdot(w(\partial^+S)-w(\partial^-S))+(w(\partial^+S)+w(\partial^-S))
\\&=4\cdot w(\partial^+S)-2\cdot w(\partial^-S)
\\&<2\cdot w(\partial^-S)-2\cdot w(\partial^-S)
\\&=0,
\end{align*}
or equivalently $\bg(S)/w(\partial S)<-1/3$, as desired.
\end{proof}

\subsection{Dynamic Data Structure}

To compute and maintain our balancing weights, we use the following data structure defined in~\cite{van2024almost}:
\begin{definition}[$\alpha$-approximate min-ratio cut data structure, Definition 5.2 of~\cite{van2024almost}]
For a weighted undirected graph $H=(V,E,\bu)$ where $\bu\in\mathbb R^E$ and $\bu(e)\in[1/U,U]$\footnote{The only difference is that~\cite{van2024almost} enforces $\bu(e)\in[1,U]$ while we require the weaker $\bu(e)\in[1/U,U]$. However, we can scale all $\bu(e)$ and $\bg(v)$ by $U$, giving $\bu(e)\in[1,U^2]$ and $\bg(v)\in[-U^2,U^2]$, which preserves all ratios and still satisfies $\log U^2=\tilde{O}(1)$. We also scale the initial $\by$ by $1/U$ and apply their data structure with each $\eta$ in \textsc{ToggleCut}$(\eta)$ scaled by $1/U$, so that the value $\bu(e)(\by(u)-\by(v))$ remains unchanged. Assuming the original constraint $\bu(e)\in[1,U]$ and tracking the scaling by $U$ is cumbersome, so we instead assume the weaker $\bu(e)\in[1/U,U]$ directly.} for $\log U=\tilde{O}(1)$, an initial gradient $\bg\in\mathbb R^V$ where $\bg\bot\mathbbm 1$ and $\bg(v)\in[-U,U]$, an initial potential vector $\by\in\mathbb R^V$, and a detection threshold parameter $\epsilon$, an $\alpha$-approximate min-ratio cut data structure $\mathcal D$ implicitly maintains the potential vector $\by$ and supports the following operations:
 \begin{itemize}
 \item $\textsc{InsertEdge}(e), \textsc{DeleteEdge}(e)$: inserts/deletes the edge $e$ with weight $\bu(e)$.
 \item $\textsc{UpdateGradient}(u,v,\delta)$: updates $\bg(u)\gets\bg(u)+\delta$ and $\bg(v)\gets\bg(v)-\delta$.
 \item $\textsc{Potential}(v)$: returns $\by(v)$.
 \end{itemize}

The data structure maintains a tuple $(g,u)$ where $g\le0$ and $u>0$ such that for some implicit cut $C\subseteq V$ we have $\bg(C)=g$ and $\bu(\partial C)\le u$, and the (negative) ratio $g/u$ is at most $1/\alpha$ times the minimum possible over all $C\subseteq V$.

The data structure has an additional operation \textsc{ToggleCut}$(\eta)$ based on the current tuple $(g,u)$ with implicit cut $C$: given a parameter $0<\eta\le1/u$, the data structure implicitly updates $\by(v)\gets\by(v)+\eta$ for all $v\in C$. Then, the data structure returns some edge set $E'$ that includes every edge $e=\{u,v\}\in E$ for which the value $\bu(e)(\by(u)-\by(v))$ has changed by at least an additive $\epsilon$ since it was inserted/last returned in $E'$.
\end{definition}

\begin{theorem}[Theorem 5.4 of~\cite{van2024almost}]\label{thm:min-ratio-cut-data-structure}
There is a deterministic $\alpha$-approximate min-ratio cut data structure for $\alpha=n^{o(1)}$ such that every operation is processed in amortized time $n^{o(1)}\log U$. Furthermore, the total number of edges returned by the algorithm after $t$ calls to \textsc{ToggleCut}$(\cdot)$ is at most $n^{o(1)}\cdot t/\epsilon$.
\end{theorem}

\section{Maximum Flow Algorithm}

In this section, we present our maximum flow algorithm. We assume access to a dynamic \emph{sparsifier} data structure which is the focus of the next section (\Cref{sec:dynamic-sparsification}).

\begin{restatable}[Dynamic sparsifier with quality $q$]{definition}{SparsifierDefinition}
For a weighted directed graph $H=(V,E,\bu)$ where $\bu\in\mathbb R^E$ and $\bu(e)\in[1/U,U]$ for $\log U=\tilde{O}(1)$, a dynamic sparsifier with quality $q$ supports the following operations:
 \begin{itemize}
 \item $\textsc{InsertArc}(e), \textsc{DeleteArc}(e)$: inserts/deletes the arc $e$ with weight $\bu(e)$.
 \item $\textsc{Sparsifier}(\beta)$: given $\beta\ge1$, returns a subgraph $\widetilde H\subseteq H$ on at most $q\cdot\beta n$ arcs such that with high probability, any cut $S\subseteq V$ with $\bu(\partial_H^+S)\ge\frac1\beta \bu(\partial_HS)$ has at least one sampled out-arc, i.e., $\partial^+_{\widetilde H}S\ne\emptyset$.
 \end{itemize}
\end{restatable}

\begin{restatable}{theorem}{SparsifierTheorem}\label{thm:sparsifier-data-structure}
There is a randomized dynamic sparsifier data structure with quality $q=n^{o(1)}\log U$ such that each call to $\textsc{InsertArc}(e), \textsc{DeleteArc}(e)$ is processed in amortized time $n^{o(1)}\log U$, and each call to $\textsc{Sparsifier}(\beta)$ takes time $O(q\beta n\log U)$.
\end{restatable}

The algorithm initializes an $\alpha$-approximate min-ratio cut data structure (\Cref{thm:min-ratio-cut-data-structure}) on the input graph, where $\by=0$ and all edges have weight $w(u,v)=\frac1{\max\{\by(v)-\by(u),0\}+1}=1$. The detection threshold parameter $\epsilon$ is fixed to $\epsilon=1/(8\alpha)$. The algorithm only updates arc/edge weights when they pass the detection threshold from the data structure. That is, whenever the data structure returns some edge set $E'$, the algorithm makes the following updates for each edge $e\in E'$ with corresponding arc $(u,v)$: first, call $\textsc{DeleteEdge}(e)$ and $\textsc{UpdateGradient}(u,v,-\bu(e))$ to remove the old contribution, and then call $\textsc{InsertEdge}(u,v,w(u,v))$ and $\textsc{UpdateGradient}(u,v,w(u,v))$, where the new weight $w(u,v)$ is computed by querying $\textsc{Potential}(u)$ and $\textsc{Potential}(v)$ to determine $\by(u)$ and $\by(v)$.

The algorithm also initializes a dynamic sparsifier data structure (\Cref{thm:sparsifier-data-structure}) with the initial arc weights of $1$. Whenever an edge weight is changed in the min-ratio cut data structure, the algorithm propagates the change to the sparsifier data structure.

To find the next augmenting path, the algorithm first computes balancing weights as follows. While the data structure maintains an implicit cut $C\subseteq V$ with $g/u\le-1/(3\alpha)$, the algorithm calls $\textsc{ToggleCut}(\eta)$ with $\eta=\frac1{16\alpha u}$. Otherwise, if $g/u>-1/(3\alpha)$, then since the data structure is $\alpha$-approximate, the minimum ratio cut is at least $-1/3$. By \Cref{lem:approx-balanced}, we have $w(\partial^+S)\ge\frac12w(\partial^-S)$ for all $S\subseteq V$. The actual weights maintained by the data structure are approximations of $w(u,v)$, but we show that they are also balanced up to constant factors.

The algorithm calls $\textsc{Sparsifier}(\beta)$ with $\beta=9$, obtaining a subgraph on $n^{1+o(1)}$ arcs. By the previous balance guarantee, there is an augmenting path from $s$ to $t$ in the sampled graph. The algorithm finds such a path in $n^{1+o(1)}$ time and updates both data structures as follows. For each arc $e=(v,u)$ along the path, the algorithm calls $\textsc{DeleteEdge}(e)$/$\textsc{DeleteArc}(e)$ for the corresponding arc/edge $(v,u)$, and then $\textsc{InsertEdge}(e)$/$\textsc{InsertArc}(u,v)$ with the new weight $\frac1{\max\{\by(v)-\by(u),0\}+1}$ of the flipped arc $(u,v)$. The algorithm also calls $\textsc{UpdateGradient}$ accordingly.

\subsection{Stability}

Since edge weights only get updated when they are returned in $E'$, the data structures only maintain approximations of the true weights. In more detail, for each arc $(u,v)$, define the real weight $w(u,v)=\frac1{\max\{\by(v)-\by(u),0\}+1}$ based on the current potential function $\by$, and define the approximate weight $\tilde w(u,v)=\bu(u,v)$ as the current weight maintained by the data structure, which is the real weight at the last time the corresponding edge $(u,v)$ was inserted or returned in $E'$.

Our algorithm requires that the approximate weights $\tilde w(u,v)$ indeed approximate the real weights. We begin with a stability guarantee on the weight function that is used throughout the analysis.
\begin{lemma}[Stability]\label{lem:stability}
Consider two values $x_1,x_2\in\mathbb R$, and define the weights $w_i=\frac1{\max\{x_i,0\}+1}$ for $i\in\{1,2\}$. For any parameter $0\le\delta<1$, if $|x_1-x_2|\le\delta/w_1$, then $\frac1{1+\delta}w_1\le w_2\le\frac1{1-\delta}w_1$.
\end{lemma}
\begin{proof}
We bound
\begin{align*}
w_2=\frac1{\max\{x_2,0\}+1}
\ge\frac1{\max\{x_1+\delta/w_1,0\}+1}
\ge\frac1{\max\{x_1,0\}+1+\delta/w_1}
=\frac1{1/w_1+\delta/w_1}
=\frac{w_1}{1+\delta}
\end{align*}
and
\begin{align*}
w_2=\frac1{\max\{x_2,0\}+1}
\le\frac1{\max\{x_1-\delta/w_1,0\}+1}
\le\frac1{\max\{x_1,0\}+1-\delta/w_1}
=\frac1{1/w_1-\delta/w_1}
=\frac{w_1}{1-\delta}
\end{align*}
as desired.
\end{proof}

\begin{lemma}[Approximation of weights]\label{lem:approx-weights}
For any arc $(u,v)$, we have $\frac1{1+\epsilon}\tilde w(u,v)\le w(u,v)\le\frac1{1-\epsilon}\tilde w(u,v)$.
\end{lemma}
\begin{proof}
Let $\tilde\by$ be the potential function when the edge $(u,v)$ was last inserted or returned in $E'$. In particular, $\tilde w(u,v)=\frac1{\max\{\tilde\by(v)-\tilde\by(u),0\}+1}$.
By the guarantee of the data structure, we have
\[ \tilde w(u,v)\cdot|(\tilde\by(u)-\tilde\by(v))-(\by(u)-\by(v))|<\epsilon \iff |(\tilde\by(u)-\tilde\by(v))-(\by(u)-\by(v))|<\epsilon/\tilde w(u,v) .\]
Applying \Cref{lem:stability} with $x_1=\tilde\by(u)-\tilde\by(v)$ and $x_2=\by(u)-\by(v)$ and $\delta=\epsilon$ finishes the proof.
\end{proof}

Finally, we show that if the algorithm stops calling $\textsc{ToggleCut}(\eta)$, at which point $g/u>-1/(3\alpha)$, then the (approximate) edge weights maintained by the data structure are balanced. In particular, $\textsc{Sparsifier}(\beta)$ with $\beta=9$ returns a subgraph containing an augmenting path from $s$ to $t$ with high probability.
\begin{lemma}
If $g/u>-1/(3\alpha)$, then $\tilde w(\partial^+S)\ge\frac19\tilde w(\partial S)$ for all $S\subseteq V$.
\end{lemma}
\begin{proof}
By \Cref{lem:approx-balanced}, we have $w(\partial^+S)\ge\frac12w(\partial^-S)$ for all $S\subseteq V$. Together with \Cref{lem:approx-weights}, we have $2\tilde w(\partial^+S)\ge w(\partial^+S)\ge\frac12w(\partial^-S)\ge\frac14\tilde w(\partial^-S)$, so $\tilde w(\partial^+S)\ge\frac18\tilde w(\partial^-S)$, or equivalently $\tilde w(\partial^+S)\ge\frac19\tilde w(\partial S)$.
\end{proof}

\subsection{Energy Bounds}\label{sec:energy-bounds}

To bound the running time, we need to track the number of calls to \textsc{ToggleCut}. Let $M=n^{O(1)}$ be large enough, and define the \emph{energy} of arc $(u,v)$ as
\[ \mathcal E(u,v)=\int_{\by(v)-\by(u)}^M\frac1{\max\{x,0\}+1}\,dx .\]
Observe that
\[ \mathcal E(u,v)=
\begin{cases}
\ln(M+1)-\ln(\by(v)-\by(u)+1)&\text{if }\by(v)\ge\by(u),
\\\ln(M+1)+\by(u)-\by(v)&\text{if }\by(v)\le\by(u).\end{cases} \]

The following lemma shows that each call to \textsc{ToggleCut} decreases the total energy by an additive $1/n^{o(1)}$.
\begin{lemma}[Energy decrease]\label{lem:energy-decrease}
On each call to \textsc{ToggleCut}, the total energy decreases by at least $\Omega(1/\alpha^3)$.
\end{lemma}
\begin{proof}
On each call to $\textsc{ToggleCut}(\eta)$, the implicit cut $C$ satisfies $g/u\le-1/(3\alpha)$, where $\bg(C)=g<0$ and $\bu(\partial C)\le u$. We also have $\bu(\partial C)\ge u/\alpha$ since otherwise, if $\bu(\partial C)<u/\alpha$, then the ratio $\bg(C)/\bu(\partial C)$ would be less than $\bg(C)/(u/\alpha)=\alpha\cdot g/u$, contradicting the fact that $g/u$ is at most $1/\alpha$ times the min-ratio cut.

We now analyze the contribution of $\textsc{ToggleCut}(\eta)$ to the total energy. Let $w$ and $\by$ be the weight and potential function immediately before $\textsc{ToggleCut}(\eta)$. For each arc $(u,v)\in\partial^+C$, the value of $\by(u)$ is increased by $\eta$, so the energy $\mathcal E(u,v)$ increases by
\[ \int_{\by(v)-\by(u)-\eta}^{\by(v)-\by(u)}\frac1{\max\{x,0\}+1}\,dx .\]
Similarly, for each arc $(u,v)\in\partial^-C$, the value of $\by(v)$ is increased by $\eta$, so the energy $\mathcal E(u,v)$ decreases by
\[ \int_{\by(v)-\by(u)}^{\by(v)-\by(u)+\eta}\frac1{\max\{x,0\}+1}\,dx .\]
By the approximation of weights (\Cref{lem:approx-weights}), we have $w(u,v)\le2\tilde w(u,v)$, so for any arc $(u,v)\in\partial C$, we have $u\ge\bu(\partial C)\ge\tilde w(u,v)\ge\frac12w(u,v)$.
Since $\eta=\frac1{16\alpha u}\le\frac1{8\alpha\cdot w(u,v)}$, by the stability guarantee (\Cref{lem:stability}) with $\delta=\frac1{8\alpha}$, we have
\[ \frac1{1+1/(8\alpha)}w(u,v)\le\frac1{\max\{x,0\}+1}\le\frac1{1-1/(8\alpha)}w(u,v) \text{ for all }x\in[\by(v)-\by(u)-\eta,\,\by(v)-\by(u)+\eta]. \]
In particular, the energy increase of each arc $(u,v)\in\partial^+C$ is at most $\eta\cdot\frac1{1-1/(8\alpha)}\cdot w(u,v)$, and the energy decrease from each arc $(u,v)\in\partial^-C$ is at least $\eta\cdot\frac1{1+1/(8\alpha)}\cdot w(u,v)$. Summed over all such arcs, the net change in energy is at most
\[ \eta\cdot\frac1{1-1/(8\alpha)}\cdot w(\partial^+C)-\eta\cdot\frac1{1+1/(8\alpha)}\cdot w(\partial^-C) .\]

Recall that $\bg(C)/\bu(\partial C)\le g/u\le-1/(3\alpha)$ and $\bg(C)=\tilde w(\partial^+C)-\tilde w(\partial^-C)$ and $\bu(\partial C)=\tilde w(\partial^+C)+\tilde w(\partial^-C)$. We have
\begin{align*}
3\alpha\cdot\bg(C)+\bu(\partial C)\le0&\iff(3\alpha+1)\cdot \tilde w(\partial^+C)-(3\alpha-1)\cdot \tilde w(\partial^-C)\le0
\\&\iff \tilde w(\partial^+C)\le\frac{3\alpha-1}{3\alpha+1}\cdot \tilde w(\partial^-C).
\end{align*}
Together with the approximation of weights (\Cref{lem:approx-weights}), we have
\[ w(\partial^+C)\le\frac1{1-\epsilon}\cdot\tilde w(\partial^+C)\le\frac1{1-\epsilon}\cdot\frac{3\alpha-1}{3\alpha+1}\cdot \tilde w(\partial^-C)\le\frac1{1-\epsilon}\cdot\frac{3\alpha-1}{3\alpha+1}\cdot(1+\epsilon)\cdot w(\partial^-C). \]
Since $\epsilon=1/(8\alpha)$, it follows that the net change in energy is at most
\begin{align*}
&\eta\cdot\frac1{1-1/(8\alpha)}\cdot w(\partial^+C)-\eta\cdot\frac1{1+1/(8\alpha)}\cdot w(\partial^-C)
\\\le{}&\eta\cdot\frac1{1-1/(8\alpha)}\cdot\frac{1+1/(8\alpha)}{1-1/(8\alpha)}\cdot\frac{3\alpha-1}{3\alpha+1}\cdot w(\partial^-C)-\eta\cdot\frac1{1+1/(8\alpha)}\cdot w(\partial^-C)
\\\le{}&\eta\cdot-\Omega\left(\frac1\alpha\right)\cdot w(\partial^-C)
\\\le{}&\eta\cdot-\Omega\left(\frac1\alpha\right)\cdot\tilde w(\partial^-C).
\end{align*}
Recall that $\bu(\partial C)\ge u/\alpha$, which means that $\eta=\frac1{16\alpha u}\ge\frac1{16\alpha^2\bu(\partial C)}$. Also, since $\bg(C)<0$, we have $\tilde w(\partial^-C)\ge\frac12\tilde w(\partial C)=\frac12\bu(\partial C)$. We conclude that the net change in energy is at most
\[ \frac1{16\alpha^2\bu(\partial C)}\cdot-\Omega\left(\frac1\alpha\right)\cdot\frac12\bu(\partial C)=-\Omega\left(\frac1{\alpha^3}\right), \]
as desired.
\end{proof}

Next, suppose that the algorithm finds an augmenting path from $s$ to $t$. The following lemma bounds the total increase in energy from flipping the arcs along the augmenting path.

\begin{lemma}[Energy increase]\label{lem:energy-increase}
Assume that all vertex potentials are within an additive $M$ of each other, i.e., $|\by(u)-\by(v)|\le M$ for all $u,v\in V$. Suppose the algorithm flips all arcs along a (simple) path from $s$ to $t$. The total energy increases by at most $\by(t)-\by(s)+(n-1)\ln(M+1)$.
\end{lemma}
\begin{proof}
Recall that
\[ \mathcal E(u,v)=
\begin{cases}
\ln(M+1)-\ln(\by(v)-\by(u)+1)&\text{if }\by(v)\ge\by(u),
\\\ln(M+1)+\by(u)-\by(v)&\text{if }\by(v)\le\by(u).\end{cases} \]
Consider the energy difference from flipping an individual arc $(u,v)$. If $\by(v)\ge\by(u)$, then the original energy is $\ln(M+1)-\ln(\by(v)-\by(u)+1)$ and the new energy is $\ln(M+1)+\by(v)-\by(u)$, so the net change in energy is
\begin{align*}
&\big(\ln(M+1)+\by(v)-\by(u)\big)-\big(\ln(M+1)-\ln(\by(v)-\by(u)+1)\big)
\\={}&\by(v)-\by(u)+\ln(\by(v)-\by(u)+1)
\\\le{}&\by(v)-\by(u)+\ln(M+1),
\end{align*}
where we use the assumption that $\by(v)-\by(u)\le M$.
Similarly, if $\by(v)\le\by(u)$, then the original energy is $\ln(M+1)+\by(u)-\by(v)$ and the new energy is $\ln(M+1)-\ln(\by(u)-\by(v)+1)$, so the net change in energy is
\begin{align*}
&\big(\ln(M+1)-\ln(\by(u)-\by(v)+1)\big)-\big(\ln(M+1)+\by(u)-\by(v)\big)
\\={}&\by(v)-\by(u)-\ln(\by(u)-\by(v)+1)
\\\le{}&\by(v)-\by(u).
\end{align*}
In both cases, the net change in energy is at most $\by(v)-\by(u)+\ln(M+1)$. Summed over all arcs $(u,v)$ on the path, the $\by(v)-\by(u)$ terms telescope to $\by(t)-\by(s)$, and the $\ln(M+1)$ terms sum to at most $(n-1)\ln(M+1)$.
\end{proof}

\subsection{Running Time}

Outside of the calls to \textsc{ToggleCut}, the algorithm spends $n^{1+o(1)}$ time on each augmenting path: the algorithm computes a sparsifier, finds an augmenting path, flips the path, and updates the data structures. To bound the number of calls to \textsc{ToggleCut}, it suffices to bound the total energy increase over the course of the algorithm, since the total energy starts at $m\ln(M+1)$ and decreases by $\Omega(\alpha^3)=1/n^{o(1)}$ on each call to \textsc{ToggleCut} (\Cref{lem:energy-decrease}).

Let $\mathcal E$ be the current total energy. By \Cref{lem:energy-increase}, the energy increase is at most $\by(t)-\by(s)+(n-1)\ln(M+1)$ after flipping the next augmenting path. If $\by(t)-\by(s)\ge0$, then we bound $\by(t)-\by(s)$ as follows. Each of the arcs $(t,s)$ added during preprocessing has energy $\ln(M+1)+\by(t)-\by(s)\ge\by(t)-\by(s)$. There are $\Omega(m)$ such arcs, so their total energy $\Omega(m)\cdot(\by(t)-\by(s))$ is at most $\mathcal E$. We conclude that $\by(t)-\by(s)\le O(\mathcal E/m)$, and the total energy after flipping the augmenting path is at most $(1+O(1/m))\cdot\mathcal E+(n-1)\ln(M+1)$. The total number of augmenting paths is $F\le m$, so the total energy increase is at most $(1+O(1/m))^F\cdot (m\ln(M+1)+F(n-1)\ln(M+1))\le O((m+nF)\ln(M+1))$.

At this point, we can also justify the assumption $\max_{u,v}|\by(u)-\by(v)|\le M$ from \Cref{lem:energy-increase} for $M=n^{O(1)}$ as long as the residual graph is strongly connected. For the sake of argument, re-define $M$ as the largest value of $\max_{u,v}|\by(u)-\by(v)|$ over the course of the algorithm, and suppose for contradiction that $M$ is too large. Consider the moment of the algorithm at which $\by(u)-\by(v)=M$ for some $u,v\in V$. Consider a path from $u$ to $v$ of at most $n-1$ arcs in the current residual graph. There is an arc $(u',v')$ on the path with $\by(u')-\by(v')\ge M/(n-1)$, which has energy $\ln(M+1)+\by(u')-\by(v')\ge M/(n-1)$. In particular, the total energy is at least $M/(n-1)$. However, by the previous energy bound, the total energy at this point is at most $O((m+nF)\ln(M+1))$. For large enough $M=n^{O(1)}$, we obtain a contradiction.

We conclude that the algorithm has total running time $(m+nF)\cdot n^{o(1)}$ while the residual graph is strongly connected. When the flow is maximum, the residual graph $G_f$ is no longer strongly connected, and the algorithm may have undefined behavior. However, we can impose a running time limit of $(m+n|f|)\cdot n^{o(1)}$ where $|f|$ is the current flow value, and terminate the algorithm if the limit is exceeded. This completes the proof of our main result (\Cref{thm:main}).

To obtain \Cref{cor:main}, we use the fact that $\sqrt n$ rounds of Dinic's blocking flow algorithm~\cite{dinic1970algorithm} reduces the maximum flow of the residual graph to at most $m/\sqrt n$~\cite{karzanov1973finding,even1975network}. Computing the blocking flows takes time $O(m\sqrt n)$, and applying \Cref{thm:main} with $F\le m/\sqrt n$ on the residual graph takes time $mn^{1/2+o(1)}$.

\section{Dynamic Sparsification}\label{sec:dynamic-sparsification}

In this section, we present our dynamic sparsifier data structure to establish \Cref{thm:sparsifier-data-structure}, restated below together with the definition.

\SparsifierDefinition*
\SparsifierTheorem*

\subsection{Reduction to Unweighted Graphs}

The dynamic algorithm first reduces to the case of unweighted graphs by bucketing the arcs and running an unweighted dynamic sparsifier with quality $n^{o(1)}$ on each bucket. More precisely, for each integer $i$, let $H_i\subseteq H$ be the subgraph of all arcs with weight in the range $[2^i,2^{i+1})$. Since $\bu(e)\in[1/U,U]$, there are only $O(\log U)$ non-empty subgraphs $H_i$. For each $H_i$, treated as a dynamic \emph{unweighted} graph, initialize a dynamic sparsifier data structure where all arcs have weight $1$. Finally, on each query to $\textsc{Sparsifier}(\beta)$, the algorithm calls $\textsc{Sparsifier}(2\beta)$ on the data structure for each $H_i$, obtaining a subgraph $\widetilde H_i\subseteq H_i$. The algorithm outputs the union of the subgraphs $\widetilde H_i$ as the final sparsifier $\widetilde H$, which has $q\cdot\beta n\cdot O(\log U)$ arcs.

To prove correctness of this reduction, consider a cut $S\subseteq V$ with $\bu(\partial_H^+S)\ge\frac1\beta \bu(\partial_HS)$, rewritten as $\beta\cdot \bu(\partial^+_HS)-\bu(\partial_HS)\ge0$. By our bucketing scheme, we can bound $\bu(\partial^+_HS)$ and $\bu(\partial_HS)$ by
\[ \bu(\partial^+_HS)\le\sum_i2^{i+1}\cdot|\partial^+_{H_i}S|\quad\text{and}\quad \bu(\partial_HS)\ge\sum_i2^i\cdot|\partial_{H_i}S|, \]
so
\[ \sum_i2^i\cdot(2\beta\cdot|\partial^+_{H_i}S|-|\partial_{H_i}S|)=\beta\cdot\sum_i2^{i+1}\cdot|\partial^+_{H_i}S|-\sum_i2^i\cdot|\partial_{H_i}S|\ge\beta\cdot \bu(\partial^+_HS)-\bu(\partial_HS)\ge0 .\]
In particular, there exists an $i$ with $2\beta\cdot|\partial^+_{H_i}S|-|\partial_{H_i}S|\ge0$, or equivalently $|\partial^+_{H_i}S|\ge\frac1{2\beta}|\partial_{H_i}S|$. Assuming that the guarantee of $\textsc{Sparsifier}(2\beta)$ on $H_i$ holds, we conclude that $\partial^+_{\widetilde H_i}S\ne\emptyset$ as well as $\partial^+_{\widetilde H}S\ne\emptyset$, completing the reduction.

\subsection{Dynamic Unweighted Sparsifier}
For the dynamic unweighted sparsifier, the algorithm maintains a sequence of explicit \emph{expander decompositions}  of the underlying undirected graph. We begin with some expander decomposition preliminaries on unweighted, undirected graphs. For a subset of vertices $S\subseteq V$, let $\textbf{\textup{vol}}(S)=\sum_{v\in S}\deg(v)$ denote the sum of vertex degrees in $S$. For a partition $\mathcal P=\{V_1,\ldots,V_k\}$ of $V$ into clusters, let $\partial\mathcal P=\cup_{i=1}^k\partial V_i$ denote the edges whose endpoints belong to different clusters.

\begin{definition}[$\phi$-expander]
An undirected graph is a \emph{$\phi$-expander} if $|\partial S|\ge\phi\cdot\min\{\textbf{\textup{vol}}(S)$, $\textbf{\textup{vol}}(V\setminus S)\}$ for all $S\subseteq V$.
\end{definition}

\begin{definition}[$\phi$-expander decomposition with slack $s$]
A $\phi$-expander decomposition of an undirected graph $G$ with slack $s\ge1$ is a partition $\mathcal P=\{V_1,\ldots,V_k\}$ of $V$ into clusters such that
 \begin{enumerate}
 \item $|\partial\mathcal P|\le\tilde O(\phi m)$, and
 \item Each induced subgraph $G[V_i]$ is a $(\phi/s)$-expander.
 \end{enumerate}
\end{definition}

We are interested in a dynamic expander decomposition algorithm with low \emph{recourse} with respect to the set $\partial\mathcal P$. We say that the algorithm has amortized recourse $\rho$ if for all $k\ge1$, the set $\partial\mathcal P$ changes by at most $\rho\cdot k$ edges after $k$ edge updates to the graph.

\begin{theorem}[Theorem~7.2 of~\cite{goranci2021expander}]\label{thm:expander-decomposition}
Let $0<\phi<1$ be a parameter. There is a randomized algorithm that, given an unweighted graph $G$ undergoing edge insertions and deletions, maintains an explicit $\phi$-expander decomposition with slack $s=2^{-O(\log^{1/2}n)}$ together with the explicit set of edges $\partial\mathcal P$. The algorithm has update time $2^{O(\sqrt{\log n})}/\phi^2$ and amortized recourse $2^{O(\sqrt{\log n})}$ with respect to $\partial\mathcal P$. The algorithm succeeds with high probability against an adaptive adversary.
\end{theorem}

We now describe the dynamic sparsification algorithm. Given the graph $H$, the algorithm first adds between $m/n$ and $2m/n$ self-loops to each vertex in $H$. This initial pre-processing has amortized recourse $O(1)$ since it takes $\Omega(m)$ edge updates before the value of $m/n$ changes by a constant factor, at which point $O(m)$ time is spent updating the self-loops. After the addition of self-loops, the graph still has $O(m)$ edges. The algorithm maintains a dynamic expander decomposition (\Cref{thm:expander-decomposition}) with parameter $\phi=2^{-\log^{3/4}n}$ on the graph with self-loops. On the subgraph of $H$ consisting of the edges $\partial\mathcal P$, the algorithm \emph{recursively} performs the same procedure, starting with adding $|\partial\mathcal P|/n$ self-loops to each vertex. Since $|\partial\mathcal P|$ is a factor $2^{\Omega(\log^{3/4}n)}$ smaller, the recursion depth is $t=O(\log^{1/4}n)$. Each level of recursion has amortized recourse $2^{O(\sqrt{\log n})}$, so the overall recourse (and update time) is $2^{O(\log^{3/4}n)}$.

Given the explicitly maintained expander decompositions, we now implement the procedure $\textsc{Sparsifier}(\beta)$. On recursion level $i\ge0$, let $G_i\subseteq H$ be the initial subgraph at level $i\ge0$ of recursion (before the addition of self-loops), and let $m_i$ be the number of edges in $G_i$. Let $\mathcal P_i$ be the expander decomposition at this level, and let $H_i=G_i\setminus \partial\mathcal P_i$ be the subgraph of edges with both endpoints in the same cluster. By construction, the input graph $H$ is the disjoint union of the subgraphs $H_i$. The algorithm independently samples each edge in $H_i$ with probability $\frac{q\beta}{2t}\cdot\frac n{m_i}$, and outputs the sparsifier $\widetilde H$ as the union of all sampled edges (over all $i$).

For the running time, observe that for each $i$, the random sampling in $H_i$ can be done in time proportional to the number of samples: first determine the number of samples $s$, which can be sampled from a binomial distribution, and then randomly sample $s$ edges from the edges in $H_i$. The expected number of samples in $H_i$ is at most $\frac{q\beta}{2t}\cdot n$, so the sparsifier $\widetilde H$ has expected size at most $q\beta n/2$. A simple concentration bound shows that the size is at most $q\beta n$ with high probability, so the running time is also $O(q\beta n)$ with high probability.

It remains to establish the sparsifier guarantee. Consider a cut $S\subseteq V$ with $|\partial_H^+S|\ge\frac1\beta |\partial_HS|$, rewritten as $\beta\cdot |\partial^+_HS|-|\partial_HS|\ge0$. Since $H$ is the disjoint union of subgraphs $H_i$, we can write $|\partial^+_HS|=\sum_i|\partial^+_{H_i}S|$ and similarly for $\partial_HS$. In particular,
\[ \sum_i(\beta\cdot |\partial^+_{H_i}S|-|\partial_{H_i}S|)=\beta\cdot\sum_i|\partial^+_{H_i}S|-\sum_i|\partial_{H_i}S|=\beta\cdot |\partial^+_HS|-|\partial_HS|\ge0 ,\]
so for some $i$, we have $\beta\cdot |\partial^+_{H_i}S|-|\partial_{H_i}S|\ge0$. We can further write $H_i$ as the disjoint union of $\phi$-expanders from the decomposition $\mathcal P_i$, and by a similar argument, there is a $\phi$-expander $X\subseteq H_i$ supported on $U\subseteq V$ with $\beta\cdot|\partial^+_X(S\cap U)|-|\partial_X(S\cap U)|\ge0$.

Our goal is to show that for a fixed $\phi$-expander $X\subseteq H_i$ supported on $U\subseteq V$, with high probability, any cut $R\subseteq U$ with $\beta\cdot|\partial^+_XR|-|\partial_XR|\ge0$ has at least one sampled out-arc in $X$. Taking a union bound over all expanders $X$ (over all levels of recursion) then handles all cuts $S\subseteq V$ in the sparsifier guarantee, completing the proof.

Consider a $\phi$-expander $X\subseteq H_i$ supported on $U\subseteq V$. We establish a ``cut counting'' style union bound by grouping the cuts $R\subseteq U$ with $\beta\cdot|\partial^+_XR|-|\partial_XR|\ge0$ by the value $\min\{|R|,|U\setminus R|\}$. For each integer $k\ge1$, there are at most $2\cdot|U|^k$ many cuts $R\subseteq U$ with $\min\{|R|,|U\setminus R|\}=k$. By the $\phi$-expander guarantee, each such cut satisfies $|\partial_XR|\ge\phi\cdot\min\{\textbf{\textup{vol}}_X(R),\textbf{\textup{vol}}_X(U\setminus R)\}$. Since at least $m_i/n$ self-loops were added to each vertex before the expander decomposition, we have $\min\{\textbf{\textup{vol}}_X(R),\textbf{\textup{vol}}_X(U\setminus R)\}\ge m_i/n\cdot\min\{|R|,|U\setminus R|\}=km_i/n$. In particular, if $\beta\cdot|\partial^+_XR|-|\partial_XR|\ge0$, then $|\partial^+_XR|\ge\frac1\beta|\partial_XR|\ge\frac{k\phi}\beta\cdot\frac{m_i}n$. Each edge in $X$ is sampled with probability $\frac{q\beta}{2t}\cdot\frac n{m_i}$, so the probability of sampling no edge in $\partial^+_XR$ is
\[ \left(1-\frac{q\beta}{2t}\cdot\frac n{m_i}\right)^{|\partial^+_XR|}\le\exp\left(-\frac{q\beta}{2t}\cdot\frac n{m_i}\cdot|\partial^+_XR|\right)\le\exp\left(-\frac{q\beta}{2t}\cdot\frac n{m_i}\cdot\frac{k\phi}\beta\cdot\frac{m_i}n\right)=\exp\left(-\frac{kq\phi}{2t}\right) .\]
Setting $q=2ct/\phi\cdot\ln n$ for an arbitrary constant $c>0$, this probability is at most $\exp(-ck\ln n)=n^{-ck}$. Taking a union bound over all cuts $R\subseteq U$ with $\min\{|R|,|U\setminus R|\}=k$, the probability of failure is at most $2\cdot|U|^k\cdot n^{-ck}\le 2n^{-(c-1)k}$ for this value of $k$. Finally, taking a union bound over all $k\ge1$, the overall probability of failure is at most $\sum_{k\ge1}2n^{-(c-1)k}\le O(n^{-(c-1)})$, as desired.

\section{Conclusion}

We conclude with some future directions. The first question is whether our weight function $w(u,v)=\frac1{\max\{\by(v)-\by(u),\,0\}+1}$ is the ``right'' function. A more natural function is $\frac1{\by(v)-\by(u)+1}$, which induces a ``barrier'' that enforces $\by(v)\ge\by(u)-1$ in the same way the push-relabel algorithm requires $d(v)\ge d(u)-1$ on arcs $(u,v)$ in the residual graph~\cite{goldberg1988new}. However, our algorithm may flip an ``upward'' arc $(v,u)$ with $\by(v)\ll\by(u)-1$, causing the new weight $\frac1{\by(v)-\by(u)+1}$ to be undefined. On the other hand, the push-relabel algorithm never sends flow along upward arcs, so perhaps the algorithm should only flip downward arcs. We leave as an open question whether our algorithm can be adapted to behave like a ``continuous'' push-relabel algorithm.

Another question is whether the balancing weights technique can be applied to other directed graph problems. Balanced directed graphs behave similarly to undirected graphs, so the method can be viewed as a way to apply the powerful techniques for undirected graphs to the more difficult setting of directed graphs. Indeed, the recent interior point method-based flow algorithms~\cite{chen2025maximum,van2024almost} all adopt this strategy to great success: each iteration reduces to a problem on an undirected graph, which is solved quickly by dynamic data structures.

Finally, one limit to the current balancing weights approach is the requirement that the graph is strongly connected. We leave as an open question whether the technique can be extended to general directed graphs. One difficulty in handling cuts with no arcs in one direction is that the arcs in the other direction may have their weights approaching $0$, which means that $\by(v)-\by(u)$ becomes arbitrarily large. One possible approach is to employ a ``softer'' weight function to control the potential differences as the weights tend to $0$.

\subsection*{Acknowledgement}

The author thanks Satish Rao for many insightful discussions.

\bibliographystyle{alpha}
\bibliography{ref}

\end{document}